\documentclass[12pt,british]{article}
\usepackage[T1]{fontenc}
\usepackage[utf8]{inputenc}
\usepackage{amsmath}
\usepackage{amsthm}
\usepackage{setspace}
\usepackage[authoryear]{natbib}
\makeatletter
\theoremstyle{plain}
\newtheorem{lem}{\protect\lemmaname}
\theoremstyle{plain}
\newtheorem{prop}{\protect\propositionname}

\usepackage{url}
\usepackage{hyperref}

\makeatother

\usepackage{babel}
\providecommand{\lemmaname}{Lemma}
\providecommand{\propositionname}{Proposition}

\begin{document}
\title{Fact-Finding in Social Networks\thanks{I thank Mikhail Drugov, Arye Hilman, Jan Fałkowski, Liu Qijun, and participants of the UC3M Micro Retreat and of the Silvaplana political economy workshop, for helpful comments and suggestions.}}
\author{Boris Ginzburg\thanks{Department of Economics, Universidad Carlos III de Madrid, Calle Madrid 126, 28903 Getafe (Madrid), Spain. Email: bginzbur@eco.uc3m.es.}}
\maketitle
\begin{abstract}
This paper models voters who invest effort to determine whether a particular claim relevant to their voting choices is correct. If a voter succeeds in determining whether the claim is correct, this information is shared via a social network. I show that increased connectivity makes voters more informed about basic facts, but less informed about complicated issues. At the same time, polarization makes voters less informed overall.
\end{abstract}

\section{Introduction}

People receive information from their friends. This also applies to politically relevant information, which is often acquired by one individual who then communicates it to his or her peers. The rise of social media has made this social channel of political communication particularly prominent. In some countries, such as the United States, social media has become the principal source of news \citep{newman2025digital}. This means that, rather than spending effort in reading the news or searching for information elsewhere, a voter can receive it from her contacts. What effects does this have on the degree to which voters are informed?

This paper proposes a simple model to study these effects. It focuses on fact-finding -- a particular type of information acquisition in which voters attempt to learn a fact, that is, to acquire a signal which conclusively reveals a state of the world relevant to their voting choice.\footnote{A particular type of fact-finding is fact-checking, when voters aim to verify whether a particular claim is accurate or not.} Acquiring the signal -- that is, finding the fact -- requires effort, which may or not lead to success. The more effort a voter invests -- for example, the more time she spends searching -- the more likely it is that she succeeds in acquiring the signal. Voters are part of a social network. If a voter succeeds in finding the fact, it is transmitted to her contacts on the network, who then transmit it to their contacts, and so on. Voters have heterogeneous preferences, which affect the value of learning the state. Each voter is also uncertain about the number of other voters to which she is directly or indirectly connected, and about their preferences.

Because a voter who acquires a signal can reveal it to her contacts, investing effort in finding a fact is akin to contributing to a public good. This public good has a particular production technology: successful information acquisition by only one voter is sufficient for all of her peers to learn the state, but this outcome depends probabilistically on the effort she invests. It turns out that the relationship between effort invested and the probability of success has a crucial effect on the results. Sometimes, voters are dealing with low-hanging fruits -- basic facts that individuals tend to learn with a low amount of effort. In other situations, the facts tend to be acquired with significant amounts of effort. This distinction has a crucial effect on the results.

The first result looks at the relationship between the degree of connectivity in the social network. If connectivity becomes greater -- for example, because more individuals join social media -- each individual tends to have more peers from who to learn the state. At the same time, because information is a public good, increased connectivity reduces the incentive to invest in acquiring a signal. Which of these effects dominates? It turns out that the answer to this question depends on the technology of information acquisition, that is, on the nature of the facts. For low-hanging fruits, increased connectivity makes voters more likely to end up learning the state. The opposite is the case with difficult problems. Thus, increased number of social media connections makes voters more informed about basic facts, but less informed about questions that tend to require more effort to answer. For example, voters may be more informed about how a particular politician voted or what a specific media outlet wrote. On the other hand, they may be less likely to know whether the politician is corrupt or whether the media outlet is biased.

The second result looks at the effect of political polarisation on the degree to which the electorate is informed. If voters tend to become more extreme in their preferences, information has less effect on their decisions. Hence, they tend to value information less. The effect of this on a voter with given preferences again depends on the type of facts. A voter with fixed preferences who faces a more polarised society ends up being more informed about difficult issues, and less informed about basic facts. At the same time, the electorate as a whole becomes less informed as a result of polarisation.

The paper contributes to the literature that studies the diffusion of information via social media. In these models, information arrives either exogenously \citep{pogorelskiy2019news,acemoglu2024model,sisak2024information}, or is supplied by strategic senders \citep{denter2021social,denter2024troll,kranton2024social,gradwohl2025social}.\footnote{There is also a large empirical literature documenting the spread of information online, principally focusing on misinformation. See \citet{zhuravskaya2020political} for an overview.} This paper adds to this literature by considering voters who decide to acquire information at a cost before sharing it.

Another literature looks at information acquisition by voters. This research typically takes two approaches. Some studies \citep{persico2004committee,martinelli2006would,matvejka2021electoral,yuksel2022specialized} look at voters who invest effort in acquiring, independently from each other, individual signals. Other papers \citep{chan2018deliberating,anesi2021policy,ginzburg2019collective,anesi2023deciding,ginzburg2025voting} consider groups who collectively decide to acquire a public signal. This paper considers a different setting: voters choose to acquire information individually, but if they succeed in acquiring it, information becomes public.

Finally, the paper is also related to the literature on threshold public goods, that is, public goods that are provided if the amount of contributions reaches a certain exogenous threshold (see \citealp{palfrey1984participation}, for a classic reference). A specific type of a threshold public good game is the volunteer's dilemma, in which contribution by only one individual is required for the success. However, unlike a classic volunteer's dilemma, here a contribution is only successful with a certain probability, which is increasing in the effort. In Section \ref{subsec:The-Volunteer-Dilemma}, I describe the relationship between my results and those of the volunteer's dilemma in more detail.

\section{Model}

There is an electorate consisting of a continuum of voters. Each voter is deciding between two options, called For and Against. To fix ideas, suppose that she is deciding between voting for or against the government. There is a state of the world $\theta\in\left\{ 0,1\right\} $. The state is unknown, and there is a common prior probability $p$ that it equals $1$. Each voter $i$ has a type $\lambda_{i}\in\left[0,1\right]$, which describes her preferences. Voting For gives her a payoff of $\theta-\lambda_{i}$. while voting Against gives a payoff of zero. Thus, $\theta$ measures the common value of voting for the government -- for example, whether the government is competent -- while the type $\lambda_{i}$ measures a voter's private preference, specifically, her ex ante reluctance to support the government. Types are drawn independently across voters from cdf $F_{\lambda}$. 

Each voter can privately invest effort $x\in\left[0,+\infty\right)$ to learn the state. The cost of doing so is $c\left(x\right)$. Given this investment, with probability $\psi\left(x\right)$ she receives a signal that perfectly reveals the state. With the remaining probability she does not receive a signal and does not update her beliefs. The functions $c$ and $\psi$ are strictly increasing and continuously differentiable. If $\psi\left(0\right)=0$ and $\lim_{x\rightarrow+\infty}\psi\left(x\right)=1$, then $\psi$ is a cdf, but in general this does not need to be the case. If $\psi$ is a cdf, we can think of $x$ as, for example, the time a voter spends searching for information, and $\psi\left(x\right)$ is then the probability that she finds the information before spending time $x$. 

I will assume that $\frac{\psi^{\prime}\left(x\right)}{c^{\prime}\left(x\right)}$ is decreasing, and $\frac{\psi^{\prime}\left(0\right)}{c^{\prime}\left(0\right)}$ is sufficiently high. This can happen, for example, if the cost of effort $c\left(\cdot\right)$ is convex, and $\psi$ is concave or ``not too convex''. This assumption ensures the existence of an interior equilibrium.


Let $h\left(x\right):=\frac{1}{c^{\prime}\left(x\right)}\frac{\psi^{\prime}\left(x\right)}{1-\psi\left(x\right)}$. Intuitively, for a given effort level $x$, the function $h\left(x\right)$ measures the marginal effect of additional effort on acquiring the signal, relative to its marginal cost, conditional on not having acquired the signal through effort $x$. It turns out that the shape of $h$ plays a key role in the results. I discuss the interpretation of $h$ in more detail in Section \ref{subsec:interpretation}.

Voters are connected to each other on a social network, through which information spreads. If a voter succeeds at acquiring the signal, she reveals it to her contacts, who then transmit this information to their contacts, and so on. For a voter $i$, let $n_{i}$ be the overall number of voter that she is connected to directly or indirectly; I will refer to this as the number of $i$'s ``eventual connections''. A voter does not know the number of her eventual connections. For a given voter, this number is drawn from a probability mass function $\left(p_{0},p_{1},...p_{N}\right)$, where $p_{n}$ is the probability that a voter is connected to $n$ other voters, and $N$ is the maximum number of eventual connections. Thus, a given voter ends up learning the state if either she or one of her eventual connections receives a signal.

The timing is as follows. First, nature draws the state. Also, it draws the type $\lambda_{i}$ of each voter, and the structure of the network, that is, the number of eventual connections of each voter. Each voter observes her type, and chooses effort $x_{i}$. With probability $\psi\left(x_{i}\right)$, each voter receives a signal, in which case she, and all of her eventual connections, learn the state. Voters then vote, and payoffs are realised.

\paragraph{Discussion of modelling choices.}

A key feature of the model is that voters receive payoffs from their actions, and not from the outcome of the vote. Hence, individuals are concerned about making a choice that ex post turns out to be right, as in models of expressive voting.\footnote{See \citet{hillman2010expressive}. For experimental evidence on expressive voting, see \citet{feddersen2009moral,shayo2012non,ginzburg2022counting}.} Note that if the set of voters is a continuum, an individual voter is pivotal with probability zero. Hence, allowing voters to also care about the collective decision is without loss of generality, as long as they also place some weight on individual decisions.

Another feature of the model is that voters do not know the number of individuals to whom they are connected. That is, a voter does not know the number of other voters who will observe her signal if she succeeds in acquiring it. Likewise, she does not know the number of voters whose signals she can observe. Instead, she knows the probability distribution of the number of connections. There can be two interpretations for this feature. First, voters may actually be unaware of the structure of the network. Second, they may know the structure of the network, but a given connection may, with some probability, fail to transmit information. For networks formed via social media, this may be because an individual may neglect to post or share a message, impeding further flow of information. In networks based on real-world interactions, two individuals may fail to encounter each other within the relevant timeframe. Either way, the probability with which a connection fails, together with the structure of the network, induces a probability distribution of the number of effective connections a given voter has.

The model also assumes that all individuals face the same cost of acquiring information. In reality, some individuals may have expert knowledge that enables them to learn the state with a given probability at a lower cost. This feature, however, is easy to incorporate in the model: as will be clear below, a voter whose cost of effort is not $c\left(x\right)$ but $\alpha c\left(x\right)$ for some $\alpha>0$ faces the same problem as a voter whose cost of effort is $c\left(x\right)$ but whose value of information is $\frac{v_{i}}{\alpha}$ instead of $v_{i}$. Hence, allowing voters to have heterogeneous costs of effort is isomorphic to perturbing the distribution of preferences, as long as these costs are proportional to $c\left(x\right)$.

Finally, the model assumes that a voter, if she is successful in acquiring the signal, shares the evidence she finds with their connections. Note that that doing so is weakly dominant under the payoff structure of this model.


\section{Analysis}

\subsection{Equilibrium}

Take a voter with type $\lambda_{i}$. Let $v_{i}$ be the difference between her payoff if she knows the state, and her payoff if her belief about the state remains at the prior. This is voter $i$'s gain from learning the state. I will refer to $v_{i}$ as the value of information for voter $i$. Let $V$ be the set of possible values of $v_{i}$, and let $F$ be the distribution of $v_{i}$ across voters. Lemma \ref{lem:v_i} in the Appendix derives $v_{i}$, $V$, and $F$.

A voter with value of information $v_{i}$ who invests effort $x$ receives an expected utility of 
\[
1-\left[1-\psi\left(x\right)\right]Qv_{i}-c\left(x\right),
\]
where $Q:=\sum_{n=0}^{N}p_{n}\left(1-\int_{w\in V}\psi\left[x\left(w\right)\right]dF\left[w\right]\right)^{n}$ is the probability that no other voter with whom $i$ has an effective connection receives a signal. Her equilibrium strategy can be described by a function $x\left(v\right)$, which determines effort choice for a given value of information $v$ and given as follows:
\begin{lem}
\label{lem:BR}A voter with value of information is $v_{i}$ exerts effort $x\left(v_{i}\right)$ given by
\begin{equation}
\psi^{\prime}\left[x\left(v_{i}\right)\right]Qv_{i}=c^{\prime}\left[x\left(v_{i}\right)\right],\label{eq:eqbm}
\end{equation}
\end{lem}
Intuitively, the left-hand side of (\ref{eq:eqbm}) is the marginal benefit from investing effort in acquiring a signal. This equals the marginal effect of effort on acquiring the signal, $\psi^{\prime}\left[x\left(v_{i}\right)\right]$, times the probability $Q$ that the signal that voter $i$ acquires is pivotal for her learning the state, times the value of knowing the state. At the interior equilibrium, this equals the marginal cost of acquiring information.


For a voter who observes value of information $v_{i}$, let $\pi\left(v_{i}\right)$ be the probability that she ends up knowing the state, either by learning it herself, or from one of her contacts. From the above, $\pi\left(v_{i}\right)=1-\left(1-\psi\left[x\left(v_{i}\right)\right]\right)Q$. Multiplying and dividing both sides of (\ref{eq:eqbm}) by $1-\psi\left[x\left(v_{i}\right)\right]$ and by $c^{\prime}\left[x\left(v_{i}\right)\right]$, we obtain the following condition that must hold at the interior equilibrium:
\begin{equation}
h\left[x\left(v_{i}\right)\right]\left[1-\pi\left(v_{i}\right)\right]v_{i}=1.\label{eq:eqbm condtion}
\end{equation}

Then the probability that a randomly selected voter ends up learning the state equals $\pi=\int_{v\in V}\pi\left(v\right)dF\left(v\right)$.

\subsection{Effect of connectivity}

Suppose voters become more connected. Specifically, suppose the distribution of the number of effective connections, changes to another distribution, which first-order stochastically dominates the initial distribution. The following result shows its effect on the share of voters that learn the state:
\begin{prop}
\label{prop:connectivity}Suppose the distribution of the number of effective connections $\left(p_{0},p_{1},...p_{N}\right)$, changes to a distribution, $\left(\tilde{p}_{0},\tilde{p}_{1},...\tilde{p}_{N}\right)$, which first-order stochastically dominates the initial distribution. Then the share of voters who learn the state decreases if $h$ is increasing, and increases if $h$ is decreasing.
\end{prop}
Intuitively, because information is a public good, an increase in connectivity causes each voter to invest less in acquiring information. As a result, each voter has more peers from whom to receive a signal, but each of these peers is less likely to acquire the signal. Which of these opposing effects is stronger? The probability that a randomly selected voter learns the state is the expectation of $\pi\left(v\right)$ over the distribution of $v$. As the connection between effort and $\pi\left(v\right)$ is given by (\ref{eq:eqbm condtion}), the effect of increased connectivity (and of the resulting decrease in effort) on the probability of a voter learning the state depends on the shape of $h$.

\subsection{Effect of polarisation}

Suppose that the distribution of preferences becomes more polarized. Recall that, a priori, the state equals one with probability $p$. Without information, voter with type $\lambda_{i}$ receives an expected payoff of $p-\lambda_{i}$ from voting For, and a payoff of zero from voting Against. A voter whose type $\lambda_{i}$ equals $p$ is indifferent at the prior belief. The further her type is from $p$, the more biased she is at the prior. We can think of voters with types further away from $p$ as being more extreme, while those whose types are close to $p$ are more centrist.

Suppose that the distribution of types $F_{\lambda}$ becomes more polarized, in the sense of the increase in the number of extreme voters. Specifically, we will say that a distribution $\hat{F}_{\lambda}$ \emph{admits higher polarisation} than distribution $F_{\lambda}$ if (i) $\hat{F}_{\lambda}\left(\lambda\right)\geq F\left(\lambda\right)$ for all $\lambda<p$, and (ii) $\hat{F}_{\lambda}\left(\lambda\right)\leq F\left(\lambda\right)$ for all $\lambda>p$. Intuitively, this means that for each voter type $\lambda$, there are fewer voters that are more centrist than her. This is a special case of a partial order introduced in \citet{ginzburg2025flexible}.\footnote{Similar single-crossing conditions on cumulative distribution functions have been used in economics literature \citep{diamond1974increases,johnson2006simple,drugov2020noise}, though not in the context of political polarization.}

We can then show that increased polarization has the following effect on voters learning the state:
\begin{prop}
\label{prop:polarisation}Suppose the distribution of types $F_{\lambda}$ changes to a different distribution of types that admits higher polarization. A voter with a given type becomes more likely to learn the state if $h$ is an increasing function, and less likely to learn the state if $h$ is a decreasing function. Overall, the share of voters who learn the state decreases.
\end{prop}
Intuitively, the value of information is higher for a voter if she is more uncertain ex ante, that is, if her preferences are such that she is closer to being indifferent at the prior belief. Polarization thus means that voters tend to have lower value of information. Specifically, increased polarization causes the distribution of the value of information $F$ to shit to a different distribution that is first-order stochastically dominated by $F$. Since the value of information decreases, a randomly selected voter invests less in information. Hence, the share of voters that end up learning the state falls. 

On the other hand, if all voters invest less in acquiring information, effort of a voter with a \emph{given} value of information $v$ becomes more pivotal for learning the state. Hence, holding $v$ fixed, a voter invests more in information (even though the shift in the distribution of $v$ causes a voter \emph{selected randomly unconditional on $v$} to invest less). The effect of increased effort $x\left(v\right)$ of a voter with a given $v$ on her probability of learning the state depends on $h$ as given by (\ref{eq:eqbm condtion}).

\section{Discussion}

\subsection{Interpretation of Results}\label{subsec:interpretation}

Propositions \ref{prop:connectivity} and \ref{prop:polarisation} say that the effects of connectivity and of polarization hinge on the shape of the function $h=\frac{1}{c^{\prime}\left(x\right)}\frac{\psi^{\prime}\left(x\right)}{1-\psi\left(x\right)}$. Intuitively, $\frac{\psi^{\prime}\left(x\right)}{1-\psi\left(x\right)}$ is the marginal probability of receiving a signal after additional effort at $x$ conditional on not receiving the signal after effort $x$. In particular, if $\psi$ is a cdf, then this is the hazard rate. Then $h$ is the ratio of this expression and the marginal cost of effort.

Suppose that $\psi$ is a cdf, and $c$ is linear. Then $h$ is proportional to the hazard rate. In this case, decreasing $h$ means a decreasing hazard rate of signal arrival. In other words, signal tends to arrive at lower effort levels. For example, if effort represents time spent searching for information, then this means that if the signal arrives, it tends to arrive soon. This happens when the state of the world represents ``a low-hanging fruit'': a simple claim that tends to be verified with a small amount of effort. More generally decreasing $h$ means that the conditional probability of receiving a signal is decreasing faster (or increasing slower) than the marginal cost. This again can be interpreted as the state representing a fact that is easy to verify. On the other hand, increasing $h$ means that information tends to arrive after a large amount of effort is invested.

Hence, Propositions \ref{prop:connectivity} and \ref{prop:polarisation} suggest that a voter facing increased connectivity or a polarizing society will be more informed about facts that are easy to verify, and less informed about facts that are hard to verify.

\subsection{Relaxing Monotonicity Requirements}\label{subsec:monotonicity}

The requirement that $\frac{\psi^{\prime}\left(x\right)}{c^{\prime}\left(x\right)}$ is decreasing may appear quite restrictive, but in fact it can be relaxed. Suppose that $\frac{\psi^{\prime}\left(x\right)}{c^{\prime}\left(x\right)}$ is instead single-peaked, that is, that there exists $x^{*}\in\left[0,+\infty\right)$ such that $\frac{\psi^{\prime}\left(x\right)}{c^{\prime}\left(x\right)}$ is increasing for $x\in\left[0,x^{*}\right)$ and decreasing for $x\in\left(x^{*},+\infty\right)$. This happens if, for example, $c$ is linear, and $\psi$ is a unimodal cdf. If the distribution of preferences is such that $v_{i}$ is sufficiently high for each voter (that is, if $\min\left\{ V\right\} $ is sufficiently high), then $\frac{\psi^{\prime}\left(0\right)}{c^{\prime}\left(0\right)}Qv_{i}>0$, so the left-hand side of (\ref{eq:eqbm}) is higher than the right-hand side at $x\left(v_{i}\right)=0$. In this case, we still have an interior equilibrium, and it can easily be shown that at this equilibrium, $\frac{\psi^{\prime}\left(x\right)}{c^{\prime}\left(x\right)}$ is decreasing in $x$. Then all the previous results hold in this more general case.

Likewise, the requirement that $h$ is monotone can be relaxed. A voter will never be interested in investing effort if its cost is higher than her value of information. Hence, the equilibrium amount of effort of voter $i$ will not exceed $c^{-1}\left(v_{i}\right)$. At the same time, if $\frac{\psi^{\prime}\left(x\right)}{c^{\prime}\left(x\right)}$ is single-peaked as described above, then a voter will never invest less effort than $x^{*}$. Hence, it is sufficient that $h$ is monotone on the interval $\left[x^{*},c^{-1}\left(v^{\max}\right)\right]$, where $v^{\max}$ is the maximum value of $v_{i}$ that is characterised in Lemma \ref{lem:v_i} in the Appendix.

\subsection{The Volunteer Dilemma}\label{subsec:The-Volunteer-Dilemma}

In a classic game known as the volunteer's dilemma, several players choose whether to make a contribution to a public good by paying some cost $b$. The public good is provided if and only if at least one player contributes. The Nash equilibrium makes two predictions, fairly robust to various perturbations of the game, about the effect of increasing the number of players. First, it reduces the probability of each player contributing. Second, it reduces the probability that \emph{any} player contributes, that is the probability that the public good is provided \citep{diekmann1985volunteer,hillenbrand2018volunteering,battaglini2025welfare}. Experimental research finds evidence for the former prediction, but typically fails to find evidence for the latter one \citep{goeree2017experimental,campos2021volunteer}.

In the model developed in this paper, investing effort in acquiring information is similar to contributing to a public good, which is provided if at least one individual in the social network of a given voter succeeds in acquiring a signal. In a standard volunteer's dilemma, the public good is provided with certainty if a player makes a contribution of $b$, and is not provided otherwise. This is equivalent to a special case of this model, in which $c\left(x\right)=x$, and $\psi\left(x\right)=0$ for $x<b$, and $\psi\left(x\right)=1$ for $x\geq b$. 

Suppose, however, that $\psi$ does not make a discrete jump around $b$, but is instead smooth. For example, a player may need to exert effort to remember to make a contribution; an experimental subject may need to spend effort to stay focused and not make a mistake. In this case, there exists a small $\varepsilon>0$ such that $\psi\left(x\right)=0$ for $x<b-\varepsilon$, $\psi\left(x\right)=1$ for $x>b+\varepsilon$, and for $x\in\left[b-\varepsilon,b+\varepsilon\right]$, the function $\psi\left(x\right)$ is initially convex and then concave. Then the discussion in Section \ref{subsec:monotonicity} suggests that the equilibrium will be on the concave section of $\psi$, at which $\psi\left(x\right)<1$. The effect of increased number of players on the probability of success can then be both positive and negative, depending on the shape of $h$ on this interval. Hence, this generalised setting can accommodate the observed experimental evidence.

\bibliographystyle{aer}
\bibliography{fact-finding}

\section*{Appendix}

Before proving the results, we can derive the value of information and its distribution.
\begin{lem}
\label{lem:v_i}The value of information for a voter $i$ is given by $v_{i}=\min\left\{ p,\lambda_{i}\right\} -p\lambda_{i}\in\left[0,p\left(1-p\right)\right]$, and its distribution across voters is given by $F\left(v_{i}\right)=F_{\lambda}\left(\frac{v_{i}}{1-p}\right)+1-F_{\lambda}\left(1-\frac{v_{i}}{p}\right)$.
\end{lem}
\begin{proof}
If voter $i$ learns the state, she votes For if $\theta=1$, and against if $\theta=0$. Her payoffs in these cases equal $1-\lambda_{i}$ and zero, respectively. Given the prior on $\theta$, her expected payoff if $p-p\lambda_{i}$. If she does not know the state, she votes For if and only if her expected payoff from doing so is positive, that is, if and only if $p-\lambda_{i}\geq0$. Hence, hear expected payoff equals $p-\lambda_{i}$ if $\lambda_{i}\leq p$, and zero otherwise. Consequently, $v_{i}=p-p\lambda_{i}-\max\left\{ p-\lambda_{i},0\right\} =\min\left\{ p,\lambda_{i}\right\} -p\lambda_{i}\in\left[0,p\left(1-p\right)\right]$.

The cdf of $v_{i}$ then equals
\[
F\left(v_{i}\right)=\Pr\left(\lambda_{i}<p\cap\lambda_{i}<\frac{v_{i}}{1-p}\right)+\Pr\left(\lambda_{i}>p\cap\lambda_{i}>1-\frac{v_{i}}{p}\right)
\]
Note that as $v_{i}\in\left[0,p\left(1-p\right)\right]$, we have $\frac{v_{i}}{1-p}\leq p$, and $1-\frac{v_{i}}{p}\geq p$. Hence,
\[
F\left(v_{i}\right)=\Pr\left(\lambda_{i}<\frac{v_{i}}{1-p}\right)+\Pr\left(\lambda_{i}>1-\frac{v_{i}}{p}\right)=F_{\lambda}\left(\frac{v_{i}}{1-p}\right)+1-F_{\lambda}\left(1-\frac{v_{i}}{p}\right).
\]
\end{proof}

\subsection*{Proof of Lemma \ref{lem:BR}.}

Suppose a player invests effort $x$. She fails to discover the state with probability $\left(1-\psi\left[x\right]\right)$. Her randomly selected contact on the social network discovers the state with probability $\int_{w\in V}\psi\left[x\left(w\right)\right]dF\left[w\right]$. If she has $n$ effective connections, then the probability that she fails to learn the state from any of them is $\left(1-\int_{w\in V}\psi\left[x\left(w\right)\right]dF\left[w\right]\right)^{n}$. Given the distribution of the number of connections, the probability that she ends up knowing the state is 
\[
1-\left(1-\psi\left[x\right]\right)\sum_{n=0}^{N}p_{n}\left(1-\int_{w\in V}\psi\left[x\left(w\right)\right]dF\left[w\right]\right)^{n}=1-\left(1-\psi\left[x\right]\right)Q.
\]
Hence, her expected utility is $\left[1-\left(1-\psi\left[x\right]\right)Q\right]v_{i}-c\left(x\right).$This is increasing in $x$ if and only if
\[
\frac{\psi^{\prime}\left[x\right]}{c^{\prime}\left[x\right]}>\frac{1}{v_{i}Q},
\]
which implies the first result.

To see that $x\left(v\right)$ is an increasing function, take any $v_{0},v_{1}\in V$ such that $v_{1}>v_{0}$. Then $\frac{1}{v_{1}Q}<\frac{1}{v_{0}Q}\iff\frac{\psi^{\prime}\left[x\left(v_{1}\right)\right]}{c^{\prime}\left[x\left(v_{1}\right)\right]}<\frac{\psi^{\prime}\left[x\left(v_{0}\right)\right]}{c^{\prime}\left[x\left(v_{0}\right)\right]}\iff x\left(v_{1}\right)>x\left(v_{0}\right)$. \qed

\subsection*{Proof of Proposition \ref{prop:connectivity}.}

Let $\tilde{x}\left(v\right)$ be the strategy under the distribution $\left(\tilde{p}_{0},\tilde{p}_{1},...\tilde{p}_{N}\right)$. Let $\tilde{Q}:=\sum_{n=0}^{N}\tilde{p}_{n}\left(1-\int_{w\in V}\psi\left[\tilde{x}\left(w\right)\right]dF\left[w\right]\right)^{n}$ be the probability that no contact of a randomly selected voter acquires the signal.

We will first show that $\tilde{x}\left(v\right)<x\left(v\right)$ for all $v\in V$. To show it, suppose the opposite holds: for some $v\in V$ we have $\tilde{x}\left(v\right)\geq x\left(v\right)$. Then
\begin{align*}
 & \tilde{x}\left(v\right)\geq x\left(v\right)\\
\iff & \frac{\psi^{\prime}\left[\tilde{x}\left(v\right)\right]}{c^{\prime}\left[\tilde{x}\left(v\right)\right]}\leq\frac{\psi^{\prime}\left[x\left(v\right)\right]}{c^{\prime}\left[x\left(v\right)\right]}\\
\iff & \tilde{Q}\geq Q\\
\implies & \frac{\psi^{\prime}\left[\tilde{x}\left(w\right)\right]}{c^{\prime}\left[\tilde{x}\left(w\right)\right]}\leq\frac{\psi^{\prime}\left[x\left(w\right)\right]}{c^{\prime}\left[x\left(w\right)\right]}\text{ for all }w\in V\\
\iff & \tilde{x}\left(w\right)\geq x\left(w\right)\text{ for all }w\in V\\
\implies & \left(1-\int_{w\in V}\psi\left[\tilde{x}\left(w\right)\right]dF\left[w\right]\right)^{n}\leq\left(1-\int_{w\in V}\psi\left[x\left(w\right)\right]dF\left[w\right]\right)^{n}\\
\implies & \tilde{Q}\leq\sum_{n=0}^{N}\tilde{p}_{n}\left(1-\int_{w\in V}\psi\left[x\left(w\right)\right]dF\left[w\right]\right)^{n}<Q
\end{align*}
which is a contradiction. Note that the last inequality follows from the fact that $\left(1-\int_{w\in V}\psi\left[x\left(w\right)\right]dF\left[w\right]\right)^{n}$ is decreasing in $n$, and hence its expected value under $\left(p_{0},p_{1},...p_{N}\right)$ (which equals $Q$) is greater than its expected value under $\left(\tilde{p}_{0},\tilde{p}_{1},...\tilde{p}_{N}\right)$, as the latter distribution first-order stochastically dominates $\left(p_{0},p_{1},...p_{N}\right)$.

Recall that the equilibrium is given by (\ref{eq:eqbm condtion}). A shift from $\left(p_{0},p_{1},...p_{N}\right)$ to $\left(\tilde{p}_{0},\tilde{p}_{1},...\tilde{p}_{N}\right)$ causes $x\left(v_{i}\right)$ to decrease for every $v_{i}$. If $h$ is an increasing function, then the left-hand side of (\ref{eq:eqbm condtion}) decreases for every $v_{i}$, and hence $\pi\left(v_{i}\right)$ must decrease to ensure that the equality holds. Hence, the probability that a randomly selected voter learns the state, which equals $\pi=\int_{v\in V}\pi\left(v\right)dF\left(v\right)$, decreases. If $h$ is a decreasing function, then the left-hand side of (\ref{eq:eqbm condtion}) increases for every $v_{i}$, so by similar logic $\pi$ increases.\qed

\subsection*{Proof of Proposition \ref{prop:polarisation}.}

Consider a shift from $F_{\lambda}$to $\hat{F}_{\lambda}$ that admits higher polarization. Let $\hat{F}\left(v_{i}\right)$ be the resulting distribution of the value of information $v_{i}$. From Lemma \ref{lem:v_i}, $F\left(v_{i}\right)=F_{\lambda}\left(\frac{v_{i}}{1-p}\right)+1-F_{\lambda}\left(1-\frac{v_{i}}{p}\right)$. Furthermore, $v_{i}\in\left[0,p\left(1-p\right)\right]$.

Note that for all $v_{i}$, we have $\frac{v_{i}}{1-p}\in\left[0,p\right]$, and $1-\frac{v_{i}}{p}\in\left[p,1\right]$. Hence, a shift $F_{\lambda}$to $\hat{F}_{\lambda}$ increases $F_{\lambda}\left(\frac{v_{i}}{1-p}\right)$ and reduces $F_{\lambda}\left(1-\frac{v_{i}}{p}\right)$. Therefore, $\hat{F}\left(v_{i}\right)\geq F\left(v_{i}\right)$ for all $v_{i}$, that is, $F$ first-order stochastically dominates $\hat{F}$.

For a voter with a given $v_{i}$, recall that $x\left(v_{i}\right)$, $Q$, and $\pi\left(v_{i}\right)$ denote her equilibrium effort level, the probability that none of her contacts receives the signal, and the probability that she ends up learning the state, respectively, under $F_{\lambda}$. Also, let $\hat{x}\left(v_{i}\right)$, $\hat{Q}$, and $\hat{\pi}\left(v_{i}\right)$ be the corresponding variables under $\hat{F}_{\lambda}$. We can show that $\hat{x}\left(v\right)>x\left(v\right)$ for all $v\in V$. To see this, suppose the opposite holds: for some $v$, we have
\begin{align*}
 & \hat{x}\left(v\right)\leq x\left(v\right)\\
\iff & \frac{\psi^{\prime}\left[\hat{x}\left(v\right)\right]}{c^{\prime}\left[\hat{x}\left(v\right)\right]}\geq\frac{\psi^{\prime}\left[x\left(v\right)\right]}{c\left[x\left(v\right)\right]}\\
\iff & \hat{Q}\leq Q\\
\iff & \frac{\psi^{\prime}\left[\hat{x}\left(w\right)\right]}{c^{\prime}\left[\hat{x}\left(w\right)\right]}\geq\frac{\psi^{\prime}\left[x\left(w\right)\right]}{c\left[x\left(w\right)\right]}\text{ for all }w\in V\\
\iff & \hat{x}\left(w\right)\leq x\left(w\right)\text{ for all }w\in V\\
\implies & \int_{w\in V}\psi\left[\hat{x}\left(w\right)\right]d\hat{F}\left[w\right]\leq\int_{w\in V}\psi\left[x\left(w\right)\right]d\hat{F}\left[w\right]<\int_{w\in V}\psi\left[x\left(w\right)\right]dF\left[w\right]\\
\implies & \hat{Q}>Q
\end{align*}
which is a contradiction. Note that the strict inequality follows from the fact that $x\left(\cdot\right)$ is an increasing function, and hence its expected value is higher under $F$ than $\hat{F}$ due to first-order stochastic dominance. The above reasoning also implies that $\hat{Q}>Q$.

Suppose that $h$ is an increasing function. Then for all $v\in V$, we have 
\[
\hat{x}\left(v\right)>x\left(v\right)\iff h\left[\hat{x}\left(v\right)\right]>h\left[x\left(v\right)\right]\iff\hat{\pi}\left(v\right)>\pi\left(v\right).
\]

On the other hand, if $h$ is a decreasing function, then for all $v\in V$, we have 
\[
\hat{x}\left(v\right)>x\left(v\right)\iff h\left[\hat{x}\left(v\right)\right]<h\left[x\left(v\right)\right]\iff\hat{\pi}\left(v\right)<\pi\left(v\right).
\]

To derive the results for the probability of a randomly selected voter learning the state, note first that $\int_{w\in V}\psi\left[\hat{x}\left(v\right)\right]d\hat{F}\left[w\right]<\int_{w\in V}\psi\left[x\left(v\right)\right]dF\left[w\right]$. To see this, assume that the opposite holds:
\begin{align*}
 & \int_{w\in V}\psi\left[\hat{x}\left(v\right)\right]d\hat{F}\left[w\right]\geq\int_{w\in V}\psi\left[x\left(v\right)\right]dF\left[w\right]\\
\implies & \sum_{n=0}^{N}p_{n}\left(1-\int_{w\in V}\psi\left[\hat{x}\left(v\right)\right]d\hat{F}\left[w\right]\right)^{n}\leq\sum_{n=0}^{N}p_{n}\left(1-\int_{w\in V}\psi\left[x\left(v\right)\right]dF\left[w\right]\right)\\
\iff & \hat{Q}\leq Q,
\end{align*}
which, as shown above, is not correct.

Note also that 
\[
\pi\left(v\right)=1-\left(1-\psi\left[x\left(v\right)\right]\right)Q=1-\left(1-\psi\left[x\left(v\right)\right]\right)\sum_{n=0}^{N}p_{n}\left(1-\int_{w\in V}\psi\left[x\left(v\right)\right]dF\left[w\right]\right)^{n+1}.
\]
 Hence, the probability that a randomly selected voter learns the state is 
\begin{align*}
\pi= & \int_{v\in V}\pi\left(v\right)dF\left(v\right)\\
= & 1-\left[\sum_{n=0}^{N}p_{n}\left(1-\int_{w\in V}\psi\left[x\left(v\right)\right]dF\left[w\right]\right)^{n}\right]\left(1-\int_{v\in V}\psi\left[x\left(v\right)\right]dF\left(v\right)\right)\\
= & 1-\left[\sum_{n=0}^{N}p_{n}\left(1-\int_{w\in V}\psi\left[x\left(v\right)\right]dF\left[w\right]\right)^{n+1}\right]
\end{align*}

As $\int_{w\in V}\psi\left[x\left(v\right)\right]dF\left[w\right]$ decreases as a result of a shift from $F$ to $\hat{F}$, it follows that $\pi$ decreases.\qed
\end{document}